\documentclass[11pt]{article}
\usepackage[margin=1in]{geometry}
\usepackage{amsmath,amsthm,amssymb,hyperref,times,mathptmx,xspace}

\makeatletter
\def\@thm#1#2#3{%
  \ifhmode\unskip\unskip\par\fi
  \normalfont
  \trivlist
  \let\thmheadnl\relax
  \let\thm@swap\@gobble
  \thm@notefont{\bfseries\upshape}
  \thm@headpunct{.}
  \thm@headsep 5\p@ plus\p@ minus\p@\relax
  \thm@space@setup
  #1
  \@topsep \thm@preskip               
  \@topsepadd \thm@postskip           
  \def\@tempa{#2}\ifx\@empty\@tempa
    \def\@tempa{\@oparg{\@begintheorem{#3}{}}[]}%
  \else
    \refstepcounter{#2}%
    \def\@tempa{\@oparg{\@begintheorem{#3}{\csname the#2\endcsname}}[]}%
  \fi
  \@tempa
}
\makeatother

\newcommand{\ghd}{\textsc{ghd}\xspace}
\newcommand{\ort}{\textsc{ort}\xspace}

\renewcommand{\b}{\{-1,1\}}
\newcommand\lambdaa\lambda

\newtheorem{proposition}{Proposition}[section]
\newtheorem{theorem}[proposition]{Theorem}
\newtheorem{fact}[proposition]{Fact}
\newtheorem{lemma}[proposition]{Lemma}
\newtheorem{claim}[proposition]{Claim}

\newtheorem{conjecture}[proposition]{Conjecture}
\theoremstyle{definition}

\theoremstyle{plain}

\DeclareMathOperator{\cb}{cb}
\DeclareMathOperator{\cost}{cost}
\DeclareMathOperator{\DD}{D}
\DeclareMathOperator{\E}{\mathbb{E}}

\DeclareMathOperator{\h}{H}
\DeclareMathOperator{\I}{I}
\DeclareMathOperator{\IC}{IC}
\DeclareMathOperator{\icost}{icost}
\DeclareMathOperator{\out}{out}

\DeclareMathOperator{\polylog}{polylog}
\DeclareMathOperator{\proj}{proj}
\DeclareMathOperator{\RR}{R}
\DeclareMathOperator{\scb}{scb}
\DeclareMathOperator{\sign}{sign}
\DeclareMathOperator{\spn}{span}

\DeclareMathOperator{\tail}{tail}
\DeclareMathOperator{\Var}{Var}

\newcommand{\eps}{\varepsilon}
\newcommand{\epsmono}{\eps\mbox{-mono}}

\newcommand{\R}{\mathbb{R}}

\newcommand{\cA}{\mathcal{A}}
\newcommand{\cB}{\mathcal{B}}

\newcommand{\cS}{\mathcal{S}}

\newcommand{\cX}{\mathcal{X}}
\newcommand{\cY}{\mathcal{Y}}
\newcommand{\cZ}{\mathcal{Z}}

\newcommand{\ceq}{\subseteq}

\newcommand{\ang}[1]{\langle{#1}\rangle}
\newcommand{\ceil}[1]{\lceil{#1}\rceil}

\newcommand{\dd}[2]{\mathrm{D}({#1}\parallel{#2})}
\newcommand{\dkl}[2]{\mathrm{D}_{\mathrm{KL}}({#1}\parallel{#2})}

\newcommand{\ip}[2]{\langle #1,#2 \rangle}
\newcommand{\lip}[2]{\left\langle #1,#2 \right\rangle}
\newcommand{\nm}[1]{\|#1\|}

\newcommand{\eat}[1]{}

\title{Information Complexity versus Corruption and\\ Applications to Orthogonality and Gap-Hamming
  \thanks{Work supported in part by NSF Grant IIS-0916565.}%
}
\author{
  Amit Chakrabarti \qquad Ranganath Kondapally \qquad Zhenghui Wang\\
  \mbox{}\\
  \normalsize Department of Computer Science, Dartmouth College\\
  \normalsize Hanover, NH 03755, USA\\
  \normalsize \{ac,\,rangak,\,zhenghui\}@cs.dartmouth.edu
}
\date{}

\begin{document}
\maketitle

\thispagestyle{empty}

\begin{abstract}

  Three decades of research in communication complexity have led to the
  invention of a number of techniques to lower bound randomized communication
  complexity. The majority of these techniques involve properties of large
  submatrices (rectangles) of the truth-table matrix defining a communication
  problem. The only technique that does not quite fit is information
  complexity, which has been investigated over the last decade.
  Here, we connect information complexity to one of the most powerful
  ``rectangular'' techniques: the recently-introduced smooth corruption (or
  ``smooth rectangle'') bound. We show that the former subsumes the latter
  under rectangular input distributions. We conjecture that this subsumption
  holds more generally, under arbitrary distributions, which would resolve the
  long-standing direct sum question for randomized communication.

  As an application, we obtain an optimal $\Omega(n)$ lower bound on the
  information complexity---under the {\em uniform distribution}---of the
  so-called orthogonality problem (ORT), which is in turn closely related to
  the much-studied Gap-Hamming-Distance (GHD).  The proof of this bound is
  along the lines of recent communication lower bounds for GHD, but we
  encounter a surprising amount of additional technical detail.

\end{abstract}



\section{Introduction}

The basic, and most widely-studied, notion of communication complexity deals
with problems in which two players---Alice and Bob---engage in a communication
protocol designed to ``solve a problem'' whose input is split between them. We
shall focus exclusively on this model here, and we shall be primarily
concerned with the problem of computing a Boolean function
$f:\cX\times\cY\to\b$. As is often the case, we are most interested in lower
bounds.

\subsection{Lower Bound Techniques and the Odd Man Out}

The preeminent textbook in the field remains that of Kushilevitz and
Nisan~\cite{KushilevitzNisan-book}, which covers the basics as well as
several advanced topics and applications. Scanning that textbook, one
finds a number of lower bounding techniques, i.e., techniques for
proving lower bounds on $\DD(f)$ and $\RR(f)$, the deterministic and
randomized (respectively) communication complexities of $f$. Some of the
more important techniques are the fooling set technique, log rank,
discrepancy and corruption.\footnote{Though the corruption technique is
discussed in Kushilevitz and Nisan, the term ``corruption'' is due to
Beame et al.~\cite{BeamePSW06}. The technique has also been called
``one-sided discrepancy'' and ``rectangle method''~\cite{Klauck03} by
other authors.} Research postdating the publication of the book has
produced a number of other such techniques, including the factorization
norms method~\cite{LinialS09}, the pattern matrix
method~\cite{Sherstov08}, the partition bound and the smooth
corruption\footnote{Jain and Klauck~\cite{JainK10} used the term
``smooth rectangle bound'', but we shall prefer the more descriptive
term ``corruption'' to ``rectangle'' throughout this article.}
bound~\cite{JainK10}.  Notably, all of these techniques ultimately boil
down to a fundamental fact called the {\em rectangle property}. One way
of stating it is that each {\em fiber} of a deterministic protocol,
defined as a maximal set of inputs $(x,y) \in \cX\times\cY$ that result
in the same communication transcript, is a combinatorial rectangle in
$\cX\times\cY$. The aforementioned lower bound techniques ultimately
invoke the rectangle property on a protocol that computes $f$; for
randomized lower bounds, (the easy direction of) Yao's minimax lemma
also comes into play.

One recent technique is an odd man out: namely, {\em information
complexity}, which was formally introduced by Chakrabarti et
al.~\cite{ChakrabartiSWY01}, generalized in subsequent
work~\cite{BarYossefJKS04,JayramKS03,BarakBCR10}, though its ideas
appear in the earlier work of Ablayev~\cite{Ablayev96} (see also Saks
and Sun~\cite{SaksS02}). Here, one defines an {\em information cost}
measure for a protocol that captures the ``amount of information
revealed'' during its execution, and then considers the resulting
complexity measure $\IC(f)$, for a function $f$. A precise definition of
the cost measure admits a few variants, but all of them quite naturally
lower bound the corresponding communication cost. The power of this
technique comes from a natural direct sum property of information cost,
which allows one to easily lower bound $\IC(f)$ for certain
well-structured functions $f$. Specifically, when $f$ is a
``combination'' of $n$ copies of a simpler function $g$, one can often scale
up a lower bound on $\IC(g)$ to obtain $\IC(f) \ge \Omega(n\,\IC(g))$.
The burden then shifts to lower bounding $\IC(g)$, and at this stage the
rectangle property is invoked, {\em but on protocols for $g$, not $f$}.

A nice consequence of lower bounding $\RR(f)$ via a lower bound on
$\IC(f)$ is that one then obtains a {\em direct sum theorem} for free:
that is, we obtain the bound $\RR(f^n) \ge \Omega(n\,\IC(f))$ as an
almost immediate corollary. We shall be more precise about this in
Section~\ref{ap-sec:prelim}.

\subsection{First Contribution: Rectangular versus Informational Methods}

It is natural to ask how, quantitatively, these numerous lower bounding
techniques relate to one another. One expects the various ``rectangular''
techniques to relate to one another, and indeed several such results are
known~\cite{Klauck03,LinialS09,JainK10}.  Here, we relate the
``informational'' technique to one of the most powerful rectangular
techniques, with respect to randomized communication complexity. To motivate
our first theorem, we begin with a sweeping conjecture.

\begin{conjecture} \label{ap-conj:main}
  The best information complexity lower bound on $\RR(f)$ is,
  asymptotically, at least as good as the smooth corruption (a.k.a.,
  smooth rectangle) bound, and hence, at least as good as the
  corruption, smooth discrepancy and discrepancy bounds.
\end{conjecture}

We point out that a very recent manuscript of Kerenidis et
al.~\cite{Kerenidis+12} claims to have settled this conjecture (for a natural
setting of parameters).  Since this work was done independent of theirs, and
due to the short interval between this writing and theirs, we shall continue
to label the statement as (our) conjecture.

In conjunction with the results of Jain and Klauck~\cite{JainK10}, the above
conjecture states that information complexity subsumes just about every other
lower bound technique for $\RR(f)$. All of these lower bound techniques
involve a choice of an input distribution. What we are able to prove is a
special case of the conjecture: the case when the input distributions
involved are rectangular.\footnote{Some authors use the term ``product
distribution'' for what we call rectangular distributions.} The statement
below is somewhat informal and neither fully detailed nor fully general: a
precise version appears as Theorem~\ref{ap-thm:scb-ic}.

\begin{theorem} \label{ap-thm:scb-ic-inf}
  Let $\rho$ be a rectangular input distribution for a communication
  problem $f:\b^n\times\b^n\to\b$. Then, with respect to $\rho$, for
  small enough errors $\eps$, the information complexity bound
  $\IC^\rho_\eps(f)$ is asymptotically as good as the smooth corruption
  bound $\scb^\rho_{400\eps,\eps}(f)$ with error parameter $400\eps$ and
  perturbation parameter $\eps$. That is, we have $\IC^\rho_\eps(f) =
  \Omega(\scb^\rho_{400\eps,\eps}(f))$.
\end{theorem}

Precise definitions of the terms in the above theorem are given in
Section~\ref{ap-sec:prelim}. We note that a recent
manuscript~\cite{BravermanW11} lower bounds information complexity by
discrepancy, a result that is similar in spirit to ours. This result is
incomparable with ours, because on the one hand discrepancy is a weaker
technique than corruption, but on the other hand there is no restriction
on the input distribution.

We remark that our proof of Theorem~\ref{ap-thm:scb-ic-inf} uses only
elementary combinatorial and information theoretic arguments, and proceeds
along intuitive lines. Accordingly, we believe that it remains of independent
interest, despite the very recent claim to a stronger result by Kerenidis et
al.~\cite{Kerenidis+12}.

\subsection{Second Contribution: Information Complexity of Orthogonality and Gap-Hamming}

The {\sc approximate-orthogonality} problem is a communication problem defined
on inputs in $\b^n\times\b^n$ by the Boolean function
\[
  \ort_{b,n}(x,y) = \begin{cases}
    1 \, , & \text{if~} |\ang{x,y}| \le b\sqrt{n} \, , \\
    -1 \, , & \text{otherwise} \, .
  \end{cases}
\]
Here, $b$ is to be thought of as a constant parameter. This problem arose
naturally in Sherstov's work on the Gap-Hamming Distance
problem~\cite{Sherstov11ghd}. This latter problem is defined as follows:
\[
  \ghd_n (x,y) = \begin{cases}
    -1, & \text{if~} \ang{x,y} \le -\sqrt{n}, \\
    1, & \text{if~}  \ang{x,y} \ge \sqrt{n}.
  \end{cases}
\]
The Gap-Hamming problem has attracted plenty of attention over the last
decade, starting from its formal introduction in Indyk and
Woodruff~\cite{IndykW03} in the context of data stream lower bounds, leading
up to a recent flurry of activity that has produced three different
proofs~\cite{ChakrabartiR11,Vidick11,Sherstov11ghd} of an optimal lower bound
$\RR(\ghd_n) = \Omega(n)$. In some recent work, Woodruff and
Zhang~\cite{WoodruffZ11} identify a need for strong lower bounds on
$\IC(\ghd)$, to be used in direct sum results. We now attempt to address such
a lower bound.

At first sight, these problems appear to be ideally suited for a lower bound
via information complexity: they are quite naturally combinations of $n$
independent communication problems, each of which gives Alice and Bob a single
input bit each. One feels that the uniform input distribution ought to be hard
for them for the intuitive reason that a successful protocol cannot afford to
ignore $\omega(\sqrt{n})$ of the coordinates of $x$ and $y$, and must
therefore convey $\Omega(1)$ information per coordinate for at least
$\Omega(n)$ coordinates.  However, turning this intuition into a formal proof
is anything but simple. 

Here, we prove an optimal $\Omega(n)$ lower bound on $\IC(\ort)$
under the uniform input distribution. This is a consequence of
Theorem~\ref{ap-thm:scb-ic-inf} above, but there turns out to be a
surprising amount of work in lower bounding $\scb(\ort)$ under the
uniform distribution. Our theorem involves the tail of the standard normal
distribution, which we denote by ``tail'':
\[
  \tail(x) := \frac{1}{\sqrt{2\pi}} \int_x^\infty e^{-x^2/2} dx \, .
\]

We also reserve $\mu$ for the uniform distribution on $\b^n\times\b^n$.

\begin{theorem} \label{ap-thm:ghd-ic}
  Let $b$ be a sufficiently large constant. Then, the corruption bound
  $\cb^{1,\mu}_{\theta}(\ort_{b,n}) = \Omega(n)$, for $\theta= \tail(2.01b)$.
  Hence, by Theorem~\ref{ap-thm:scb-ic-inf}, we have
  $\IC^\mu_{\theta/400}(\ort_{b,n}) = \Omega(n)$.
\end{theorem}

Again, precise definitions of the terms in the above theorem are given in
Section~\ref{ap-sec:prelim} and the proof of the theorem appears in
Section~\ref{ap-sec:ghd}.  As it turns out, a slight strengthening of the
parameter $\theta$ in the above theorem would give us the result
$\IC^\mu_{\theta'}(\ghd_n) = \Omega(n)$. This is because the following
result---stated somewhat imprecisely for now---connects the two problems.

\begin{theorem} \label{ap-thm:ghd-ort}
  Let $b$ be a sufficiently large constant and let $\theta=\tail(1.99b)$.
  Then, we have $\scb^{\mu}_{400\theta, \theta}(\ghd_n) =
  \Omega(\cb_{400\theta}^{1,\mu}(\ort_{b,n}))-O(\sqrt n)$.  By
  Theorem~\ref{ap-thm:scb-ic-inf}, we then have $\IC^{\mu}_
  \theta(\ghd_n)=\Omega(\cb_{400\theta}^{1,\mu}(\ort_{b,n})) -O(\sqrt n)$.
\end{theorem}


We note that Chakrabarti and Regev~\cite{ChakrabartiR11} state that their
lower bound technique for $\RR(\ghd_n)$ can be captured within the smooth
rectangle bound framework. While this is true in spirit, there is a
significant devil in the details, and their technique does not yield a good
lower bound on $\scb^\mu_{\eps,\delta}(\ghd_{n})$ for the {\em uniform}
distribution $\mu$. We explain more in Section~\ref{ap-sec:ghd}.

These theorems suggest a natural follow-up conjecture that we leave open.

\begin{conjecture} \label{conj:ghd-ic}
  There exists a constant $\eps$ such that $\IC^\mu_\eps(\ghd_n) = \Omega(n)$.
\end{conjecture}

\subsection{Direct Sum}

A direct sum theorem states that solving $m$ independent instances of a
problem requires about $m$ times the resources that solving a single instance
does. It could apply to a number of models of computation, with ``resources''
interpreted appropriately. For our model of two-party communication, it works
as follows.  For a function $f:\cX\times\cY\to\b$, let
$f^m:\cX^m\times\cY^m\to\b^m$ denote the function given by
\[
  f^m(x_1,\ldots, x_m, y_1,\ldots, y_m) 
  ~=~ (f(x_1,y_1), \ldots, f(x_m,y_m)) \, .
\]
Notice that $f^m$ is not a Boolean function. We will define $\RR(f^m)$
to be the randomized communication complexity of the task of outputting
a vector $(z_1,\ldots,z_m)$ such that for each $i\in [m]$, we have
$f(x_i,y_i) = z_i$ with high probability.  Then, a direct sum theorem for
randomized communication complexity would say that $\RR(f^m) =
\Omega(m\cdot\RR(f))$. Whether or not such a theorem holds for a general
$f$ is a major open question in the field. 

Information complexity, by its very design, provides a natural approach
towards proving a direct sum theorem. Indeed, this was the original
motivation of Chakrabarti et al.~\cite{ChakrabartiSWY01} in introducing
information complexity; they proved a direct sum theorem for randomized
{\em simultaneous-message} and {\em one-way} complexity, for functions
$f$ satisfying a certain ``robustness'' condition. Still using
information complexity, Jain et al.~\cite{JainRS03icalp} proved a direct
sum theorem for bounded-round randomized complexity, when $f$ is hard
under a product distribution. Recently, Barak et al.~\cite{BarakBCR10}
used information complexity, together with a {\em protocol compression}
approach, to mount the strongest attack yet on the direct sum question
for $\RR(f)$, for fairly general $f$: they show that $\RR(f^m) \approx
\Omega(\sqrt{m}\cdot\RR(f))$, where the ``$\approx$'' ignores
logarithmic factors.

One consequence of our work here is a simple proof of a direct sum
theorem for randomized communication complexity for functions whose
hardness is captured by a smooth corruption bound (which in turn
subsumes corruption, discrepancy and smooth discrepancy~\cite{JainK10})
under a rectangular distribution. This includes the well-studied
\textsc{inner-product} function, and thanks to our
Theorem~\ref{ap-thm:ghd-ic}, it also includes $\ort$.  Should
Conjecture~\ref{ap-conj:main} be shown to hold, we could remove the
rectangularity constraint altogether and capture additional important
functions such as \textsc{disjointness}, whose hardness seems to be
captured only by considering corruption under a non-rectangular
distribution. 

We note that the protocol compression approach~\cite{BarakBCR10} gives a
strong direct sum result for distributional complexity under rectangular
distributions, but still not as strong as ours because their result
contains a not-quite-benign polylogarithmic factor. We say more about
this in Section~\ref{ap-sec:ghd}.

\paragraph{Comparison with Direct Product.}
Other authors have considered a related, yet different, concept of
direct {\em product} theorems. A strong direct product theorem (henceforth,
SDPT) says that computing $f^m$ with a correctness probability as small
as $2^{-\Omega(m)}$---but more than the trivial guessing
bound---requires $\Omega(m\,\RR(f))$ communication, where
``correctness'' means getting {\em all} $m$ coordinates of the output
right. It is known that SDPTs do not hold in all
situations~\cite{Shaltiel03}, but do hold for (generalized)
discrepancy~\cite{LeeSS08,Sherstov11dirprod}, an especially important
technique in lower bounding quantum communication.  A recent manuscript
offers an SDPT for bounded-round randomized
communication~\cite{JainPY12}.

Although strong direct product theorems appear stronger than direct sum
theorems,\footnote{Some authors interpret ``direct sum'' as requiring
correctness of the entire $m$-tuple output with high probability. Under
this interpretation, direct product theorems indeed subsume direct sum
theorems. Our definition of direct sum is arguably more natural, because
under our definition, we at least have $\RR(f^m) = O(m\,\RR(f))$
always.} they are in fact incomparable. A protocol could conceivably
achieve low error on each coordinate of
$f^m(x_1,\ldots,x_m,y_1,\ldots,y_m)$ while also having zero probability
of getting the entire $m$-tuple right.

\section{Preliminaries} \label{ap-sec:prelim}

Consider a function $f:\cX\times\cY\to\cZ$, where $\cX,\cY,\cZ$ are
nonempty finite sets. Although we will develop some initial theory under
this general setting, it will be useful to keep in mind the important
special case $\cX = \cY = \b^n$ and $\cZ = \b$. We can interpret such a
function $f$ as a {\em communication problem} wherein Alice receives an
input $x\in\cX$, Bob receives an input $y\in\cY$, and the players must
communicate according to a {\em protocol} $P$ to come up with a value
$z\in\cZ$ that is hopefully equal to $f(x,y)$. The sequence of messages
exchanged by the players when executing $P$ on input $(x,y)$ is called
the {\em transcript} of $P$ on that input, and denoted $P(x,y)$. We
require that the transcript be a sequence of bits, and end with (a
binary encoding of) the agreed-upon output. We denote the output
corresponding to a transcript $t$ by $\out(t)$: thus, the output of $P$
on input $(x,y)$ is $\out(P(x,y))$.

Our protocols will, in general, be randomized protocols with a public
coin as well as a private coin for each player. When we disallow the
public coin, we will explicitly state that the protocol is private-coin.
Notice that $P(x,y)$ is a random string, even for a fixed input $(x,y)$.
For a real quantity $\eps \ge 0$, we say that $P$ computes $f$ with
$\eps$ error if $\Pr[\out(P(x,y)) \ne f(x,y)] \le \eps$, the probability
being with respect to the randomness used by $P$ and the input distribution. 
We define the cost of $P$ to be the worst case length of its transcript, 
$\max |P(x,y)|$, where we maximize over all inputs $(x,y)$ and over all 
possible outcomes of the coin tosses in $P$. Finally, the $\eps$-error 
randomized communication complexity of $f$ is defined by
\[
  \RR_\eps(f) = \min\{\cost(P):\, P~\text{computes}~f~\text{with
  error}~\eps\} \, .
\]
In case $\cZ = \b$, we also put $\RR(f) = \RR_{1/3}(f)$.

For random variables $A,B,C$, we use notations of the form $\h(A)$,
$\h(A\mid C)$, $\h(AB)$, $\I(A:B)$, and $\I(A:B\mid C)$ to denote
entropy, conditional entropy, joint entropy, mutual information, and
conditional mutual information respectively. For discrete probability
distributions $\lambda,\mu$, we use $\dkl{\lambda}{\mu}$ to denote the
relative entropy (a.k.a., informational divergence or Kullback-Leibler
divergence) from $\lambda$ to $\mu$ using logarithms to the base $2$.
These standard information theoretic concepts are well described in a
number of textbooks, e.g., Cover and Thomas~\cite{CoverThomas-book}.

Let $\lambda$ be an input distribution for $f$, i.e., a probability
distribution on $\cX\times\cY$. We say that $\lambda$ is a {\em
rectangular distribution} if we can write it as a tensor product
$\lambda = \lambda_1 \otimes \lambda_2$, where $\lambda_1,\lambda_2$ are
distributions on $\cX,\cY$ respectively. Now consider a general
$\lambda$ and let $(X,Y) \sim \lambda$ be a random input for $f$ drawn
from this joint distribution. We define the $\lambda$-information-cost
of the protocol $P$ to be $\icost^\lambda(P) = \I(XY:P(X,Y) \mid R)$,
where $R$ denotes the public randomness used by $P$. This cost measure
gives us a different complexity measure called the $\eps$-error {\em
information complexity} of $f$, under $\lambda$:
\[
  \IC^\lambda_\eps(f) = \min\{\icost^\lambda(P):\,
  P~\text{computes}~f~\text{with error}~\eps\} \, .
\]
We note that in the terminology of Barak et al.~\cite{BarakBCR10}, the
above quantity would be called the {\em external} information
complexity, as opposed to the {\em internal} one, which is based on the
cost function $\I(X:P(X,Y),R \mid Y) + \I(Y:P(X,Y),R \mid X)$. As noted
by them, the two cost measures coincide under a rectangular input
distribution. Since our work only concerns rectangular distributions,
this internal/external distinction is not important to us.

It is easy to see (and by now well-known) that information complexity
under {\em any} input distribution lower bounds randomized communication
complexity.

\begin{fact} \label{ap-fact:ic-comm}
  For every input distribution $\lambda$ and error $\eps$, we have 
  $\RR_\eps(f) \ge \IC^\lambda_\eps(f)$.
\end{fact}
\begin{proof}
  Simply observe that $\I(XY:P(X,Y) \mid R) \le \h(P(X,Y)) \le
  |P(X,Y)|$.
\end{proof}

\subsection{Corruption and Smooth Corruption}

We consider a communication problem given by a partial function,
$f:\cX\times\cY\to\cZ\cup\{*\}$. We say that the function $f$ is undefined on
an input $(x,y)\in\cX\times\cY$ iff $f(x,y)=*$. For such inputs we say that 
a protocol $P$ computes $f$ correctly on $(x,y)$ always, irrespective of 
what $P$ outputs. Therefore, we say that a protocol $P$ computes $f$ with 
error $\eps\ge0$ if $\Pr[f(x,y)\not=*\ \wedge\ out(P(x,y))\not=f(x,y)] \le \eps$
where, as before, the probability being with respect to the randomness used by $P$ and the 
input distribution.

Pick a particular $z\in\cZ$. A set
$S\ceq\cX\times\cY$ is said to be {\em rectangular} if we have $S = S_1
\times S_2$, where $S_1\ceq\cX,S_2\ceq\cY$. Following Beame et
al.~\cite{BeamePSW06}, we say that $S$ is $\eps$-error $z$-monochromatic
for $f$ under $\lambda$ if $\lambda(S\setminus (f^{-1}(z)\cup f^{-1}(*))) \le
\eps\,\lambda(S)$. We then define
\begin{align}
  \epsmono^{z,\lambda}(f) &= \max\{\lambda(S):\, S~\mbox{is rectangular
    and $\eps$-error $z$-monochromatic}\} \, , \label{ap-def:eps-mono} \\
  \cb^{z,\lambda}_\eps(f) &= -\log (\epsmono^{z,\lambda}(f)) \, , 
    \label{ap-def:cb} \\
  \scb^{z,\lambda}_{\eps,\delta}(f) &= \max\{ \cb^{z,\lambda}_\eps(g):\,
    g\in(\cZ\cup \{*\})^{\cX\times\cY}, \Pr_{(X,Y)\sim\lambda}[f(X,Y) \ne g(X,Y)] 
      \le \delta \} \, .  \label{ap-def:scb}
\end{align}
The quantities $\cb^{z,\lambda}_\eps(f)$ and
$\scb^{z,\lambda}_{\eps,\delta}(f)$ are called the corruption bound and
the smooth corruption bound respectively, under the indicated choice of
parameters. In the latter quantity, we refer to $\eps$ as the {\em error
parameter} and $\delta$ as the {\em perturbation parameter}. One can go
on to define bounds independent of $z$ and $\lambda$ by appropriately
maximizing over these two parameters, but we shall not do that here.

We note that Jain and Klauck~\cite{JainK10} use somewhat different
notation: what we have called $\scb$ above is the logarithm of (a slight
variant of) the quantity that they call the ``natural definition of the
smooth rectangle bound'' and denote $\widetilde{srec}$.

What justifies calling these quantities ``bounds'' is that they can be
shown to lower bound $\RR_{\eps'}(f)$ for sufficiently small
$\delta,\eps,\eps'$, under a mild condition on $\lambda$.  It is clear that
$\scb^{z,\lambda}_{\eps,\delta}(f) \ge \cb^{z,\lambda}_\eps(f)$, so we
mention only the stronger result, that involves the smooth corruption
bound.

\begin{fact}[Jain and Klauck~\cite{JainK10}] \label{ap-fact:scb}
  Let $f:\cX\times\cY\to\cZ\cup\{*\}$, $z\in\cZ$ and distribution $\lambda$ on
  $\cX\times\cY$ be such that $\lambda(f^{-1}(z)) \ge 1/3$. Then there
  is an absolute constant $c > 0$ such that, for a sufficiently small
  constant $\eps$, we have $\RR_\eps(f) \ge
  c\cdot\scb^{z,\lambda}_{2\eps,\eps/2}(f)$. \qed
\end{fact}

The constant $1/3$ above is arbitrary and can be parametrized, but we
avoid doing this to keep things simple. The proof of the above fact is
along the expected lines: an application of (the easy direction of)
Yao's minimax lemma, followed by a straightforward estimation argument
applied to the rectangles of the resulting deterministic protocol. Note
that we never have to involve the linear-programming-based smooth
rectangle bound as defined by Jain and Klauck.

\section{Information Complexity versus Corruption} \label{ap-sec:rect}

We are now in a position to tackle our first theorem.

\begin{theorem}[Precise restatement of Theorem~\ref{ap-thm:scb-ic-inf}] \label{ap-thm:scb-ic}
\label{ap-thm:IC-corr-rel}
  Suppose we have a function $f: \cX \times \cY \to \cZ\cup\{*\}$, a rectangular
  distribution $\rho$ on $\cX \times \cY$, and $z\in\cZ$ satisfying
  $\rho(f^{-1}(z)) \ge 3/20$. Let $\eps, \eps'$ be reals with $0 \le 384\eps
  \le \eps' < 1/4$. Then
  \[ 
    \IC^\rho_\eps(f) \ge \frac1{400} \scb^{z,\rho}_{\eps',\eps}(f) - \frac1{50}
    = \Omega\big( \scb^{z,\rho}_{\eps',\eps}(f) \big) - O(1) \, .
  \]
\end{theorem}

To prove this, we first consider a notion that we call the {\em distortion} of
a transcript of a communication protocol. Let $\rho$ be an input
distribution for a communication problem, let $P$ be a protocol for the
problem, and let $t$ be a transcript of $P$. We define $\sigma_t =
\sigma_t(\rho)$ to be the distribution $(\rho \mid P(X,Y) = t)$. We think
of the relative entropy $\dkl{\sigma_t}{\rho}$ as a distortion measure for
$t$: intuitively, if $t$ conveys little information about the inputs, then
this distortion should be low. The following lemma makes this intuition
precise. Notice that it does not assume that $\rho$ is rectangular.

For the remainder of this section, to keep the notation simple while handling 
partial functions, we write $g(x,y)\not=z$ to actually denote the event 
$g(x,y)\not= z \wedge g(x,y)\not= *$ for $z\in\cZ$, unless specified otherwise.

\begin{lemma} \label{ap-lem:low-distortion}
  Let $P$ be a private-coin protocol that computes $g:\cX\times\cY\to\cZ\cup\{*\}$ with
  error $\eps < 1/500$. Let $z \in \cZ$ and let $\rho$ be an arbitrary
  distribution on $\cX \times \cY$ with $\rho(g^{-1}(z)) \ge 3/20 - 1/500$.  Then, there
  exists a (``low-distortion'') transcript $t$ of $P$ such that
  \begin{align}
    \out(t) &= z \, , \label{ap-cond:tran-output} \\
    \dkl{\sigma_t}{\rho} &\le 50 \icost^\rho(P) \, , \text{~~and} \label{ap-cond:low-divergence} \\
    \Pr[g(X,Y) \ne z \mid T=t] &\le 8\eps \, , \label{ap-cond:low-error}
  \end{align}
  where $(X,Y)\sim\rho$ and $T = P(X,Y)$.
\end{lemma}
\begin{proof}
  Let $\tau$ denote the distribution on transcripts given by $P(X,Y)$.
  By basic results in information theory~\cite{CoverThomas-book}, we
  have
  \[
    \icost^\rho(P) = \I(XY:T) = \E_{T\sim\tau}\left[
      \dkl{\sigma_T}{\rho}\right] \, .
  \]
  Consider a random choice of $t$ according to $\tau$. By Markov's inequality,
  conditions~\eqref{ap-cond:low-divergence} and~\eqref{ap-cond:low-error} fail with
  probability at most $1/50$ and $1/8$ respectively.  By the lower bound on $\rho(g^{-1}(z))$,
  condition~\eqref{ap-cond:tran-output} fails with probability at most
  $17/20 + 1/500+ \eps$. Since $\eps \le 1/500$, and $1/8 + 1/50 + 17/20 + 1/500 + 1/500 <
  1$, it follows that there exists a choice of $t$ satisfying all three
  conditions.
\end{proof}

Property~\ref{ap-cond:low-error} in the above lemma should be interpreted as a
low-error guarantee for the transcript $t$. We now argue that the existence of
such a transcript implies the existence of a ``large'' low-corruption
rectangle, provided the input distribution $\rho$ is rectangular: this is the
only point in the proof that uses rectangularity. One has to be careful with
the interpretation of ``large'' here: it means large under $\sigma_t$, and not
$\rho$. However, later on we will add in the low-distortion guarantee of
Lemma~\ref{ap-lem:low-distortion} to conclude largeness under $\rho$ as well.

\begin{lemma}
  Let $t$ be a transcript of a private-coin protocol $P$ for
  $g:\cX\times\cY\to\cZ\cup\{*\}$. Let $\rho$ be a rectangular distribution on
  $\cX\times\cY$, $z\in\cZ$, $(X,Y)\sim\rho$, $T = P(X,Y)$, and $\eps \ge
  0$. Suppose
  \begin{equation} \label{ap-cond:low-err-rho-app}
    \Pr[g(X,Y) \ne z \mid T=t] \leq \eps \, ,
  \end{equation}
  then there exists a rectangle $L\ceq\cX\times\cY$ such that 
  \begin{align}
    \sigma_t(L) &\geq 9/16 \, , ~\text{and} \label{ap-cond:large-under-rho-app}\\
    \Pr[g(X,Y) \ne z \mid (X,Y) \in L] &\leq 16\eps \, .  \label{ap-cond:low-err-mu-app}
  \end{align}
\label{lem:app-main}
\end{lemma}
\begin{proof}
  By the rectangle property for private-coin
  protocols~\cite[Lemma~6.7]{BarYossefJKS04}, there exist mappings
  $q_1:\cX\to [0,1], q_2:\cY\to [0,1]$ such that $\Pr[T=t \mid X=x, Y=y]
  = q_1(x) q_2(y)$. 


  Let $\tau$ denote the distribution of $T$. We can rewrite the
  condition~\eqref{ap-cond:low-err-rho-app} as
  \begin{equation}
    \sum_{x \in \cX, y \in \cY : g(x,y) \ne z} q_1(x) q_2(y) 
    \rho(x,y) \le \eps\,\tau(t) \label{ap-eq:err-mass-t-app} \, .
  \end{equation}

  Consider the set $\cA$ of rows whose contribution to the left hand
  side of~\eqref{ap-eq:err-mass-t-app} is ``low,'' i.e.,
  \[
    \cA = \bigg\{x \in \cX:\, \sum_{y:g(x,y) \ne z} q_2(y) \rho(x,y)
      \le 4\eps \sum_y q_2(y) \rho(x,y) \bigg\} \, .
  \]
  Then, by a Markov-inequality-style argument, we have $\Pr[ X \in \cA
  \mid T=t] \ge \frac34$.


  Similarly, consider the following set $\cB$ of columns (notice that we
  sum over only $x\in\cA$):
  \[
    \cB = \bigg\{y \in \cY:\, \sum_{x \in \cA: g(x,y) \ne z} \rho(x,y) 
      \le 16\eps \sum_{x \in \cA} \rho(x,y) \bigg\} \, .
  \]
  We now claim that the rectangle $\cA \times \cB$ has the desired
  properties.

  From the definition of $\cB$, it follows that for all $y\in\cB$,
  $\Pr[g(X,y)\not=z\mid X\in A] \le 16\eps$.  Therefore, we have $\Pr[g(X,Y)
  \ne z \mid (X,Y)\in \cA\times\cB] \le 16\eps$ and hence, the rectangle
  $\cA\times\cB$ satisfies condition~\eqref{ap-cond:low-err-mu-app}.
 
  \eat{
  It is easy to see that $\cA \times \cB$ satisfies
  condition~\eqref{ap-cond:low-err-mu-app}: 
  \begin{align*}
  \Pr[g(X,Y) \ne z \mid (X,Y) \in \cA \times \cB] 
  &= \left( \sum_{x \in \cA}\sum_{y \in \cB : g(x,y) \ne z} \rho(x,y) \right) / 
     \left( \sum_{x \in \cA}\sum_{y \in \cB}  \rho(x,y) \right) \\
  &= \left( \sum_{y \in \cB}\sum_{x \in \cA : g(x,y) \ne z}  \rho(x,y) \right) / 
     \left( \sum_{x \in \cA}\sum_{y \in \cB}  \rho(x,y) \right) \\
  &\le  \left( \sum_{y \in \cB} 16\eps \sum_{x \in \cA}  \rho(x,y) \right) / 
     \left( \sum_{x \in \cA}\sum_{y \in \cB}  \rho(x,y) \right) \\
  &= 16\eps
  \end{align*}
  }

Since we know that $\Pr[X\in\cA\mid T=t] \ge 3/4$, to prove that $\Pr[(X,Y)\in \cA\times\cB \mid T=t]\ge 9/16$ 
we will first show that the columns in $\cB$ have significant ``mass'' in $\cA$ using averaging arguments.
\begin{claim} We have
 $\sum_{x \in \cA}\sum_{y \in \cB} q_2(y) \rho(x,y) \geq \frac 34 \sum_{x \in \cA} 
   \sum_{y\in \cY} q_2(y) \rho(x,y)$.
\label{ap-clm:mass-B}
\end{claim}
\begin{proof}
Assume not. Then $\sum_{x \in \cA}\sum_{y \in \cY \setminus \cB} q_2(y) \rho(x,y) 
\geq \frac14 \sum_{x \in \cA} \sum_{y\in\cY} q_2(y) \rho(x,y)$. Therefore,
\begin{eqnarray*}
 \sum_{y\in\cY} \sum_{x \in \cA : g(x,y) \ne z} q_2(y)  \rho(x,y)
  &\geq& \sum_{y \in \cY \setminus \cB} q_2(y) \sum_{x \in \cA : g(x,y) \ne z} \rho(x,y)\\
  &>& 16\eps \sum_{y \in \cY \setminus \cB} q_2(y) \sum_{x \in \cA} \rho(x,y)
   \quad (\text{by def of $\cB$})\\
  &\geq& 4\eps \sum_{y\in \cY} q_2(y) \sum_{x \in \cA} \rho(x,y) \, ,
\end{eqnarray*}
which contradicts the definition of $\cA$.
\end{proof}

Recall that $\rho$ is a rectangular distribution. Suppose $\eta_1$ and
$\eta_2$ are its marginals, i.e., $\rho(x,y) = \eta_1(x) \eta_2(y)$.  We now
observe that the fraction $\sum_{y \in \cB}q_2(y) \rho(x,y)/
\sum_{y\in\cY} q_2(y) \rho(x,y)$ is the same for all $x\in\cX$. We have
\[
  \frac{\sum_{y \in \cB}q_2(y)  \rho(x,y)} {\sum_{y\in\cY} q_2(y)  \rho(x,y)}
    = \frac{\sum_{y \in \cB}q_2(y)  \eta_1(x)  \eta_2(y)} {\sum_{y\in\cY} q_2(y)  \eta_1(x)  \eta_2(y)}
    = \frac{\eta_1(x)\sum_{y \in \cB}q_2(y)  \eta_2(y)} {\eta_1(x)\sum_{y\in\cY} q_2(y)  \eta_2(y)}
    = \frac{\sum_{y \in \cB}q_2(y)  \eta_2(y)} {\sum_{y\in\cY} q_2(y)  \eta_2(y)} \, ,
\]
which is indeed independent of $x$. Denote this fraction by $\kappa$.
With the above observation and claim~\ref{ap-clm:mass-B}, we can conclude that $\kappa \ge 3/4$.
We can now prove that the rectangle $\cA \times \cB$ satisfies condition~\eqref{ap-cond:large-under-rho-app}
as follows:
\begin{align*}
\sigma_t(\cA \times \cB )
 &= \sum_{x \in \cA} \sum_{y \in \cB} \rho(x,y)  q_1(x)  q_2(y) /\tau(t)\\
 &= \sum_{x \in \cA} \eta_1(x)  q_1(x) \sum_{y \in \cB} q_2(y)  \eta_2(y) /\tau(t)\\
 &= \sum_{x \in \cA} \eta_1(x)  q_1(x) \kappa \sum_{y\in\cY} q_2(y)  \eta_2(y) /\tau(t)\\
 &= \kappa \sum_{x \in \cA, y\in\cY} \eta_1(x)  q_1(x)  q_2(y)  \eta_2(y) /\tau(t)\\
 &= \kappa \sum_{x \in \cA, y\in\cY} \rho(x,y)  q_1(x)  q_2(y) /\tau(t)\\
 &= \kappa \Pr[X\in\cA\mid T=t] \ge \frac{3\kappa}{4} \ge \frac{9}{16} \, .
\qedhere
\end{align*}

\end{proof}

The proof of our next lemma uses the (classical) Substate Theorem due to
Jain, Radhakrishnan and Sen~\cite{JainRS09}. We state this below in a
form that is especially useful for us: it says roughly that if the
relative entropy $\dkl{\lambdaa_1}{\lambdaa_2}$ is upper bounded, then
the events that have significant probability under $\lambdaa_1$ continue to
have significant probability under $\lambdaa_2$.

\begin{fact}[Substate Theorem~\cite{JainRS09}] \label{ap-fact:substate}
  Let $\lambdaa_1$ and $\lambdaa_2$ be distributions on a set
  $\cX$ with $\dkl{\lambdaa_1}{\lambdaa_2} \le d$, for some positive $d$.
  Then, for all $S\ceq\cX$, we have $\lambdaa_2(S) \ge
  \lambdaa_1(S)/2^{2 + 2/\lambdaa_1(S) + 2d/\lambdaa_1(S)}$. \qed
\end{fact}

\begin{lemma} \label{ap-lem:no-t-with-three}
  Let $t$ be a transcript of a private-coin protocol $P$ for
  $g:\cX\times\cY\to\cZ\cup\{*\}$, and suppose $\out(t) = z \in \cZ$.  Let $\rho$ be a
  rectangular distribution on $\cX \times \cY$, and $\eps \le 1$.  Then at
  most one of the following conditions can hold:
  \begin{align}
    \dkl{\sigma_t}{\rho} &< (\cb^{z,\rho}_\eps(g) - 7) / 4\, , \label{ap-prp:low-kl-div}\\
    \Pr[g(X,Y) \ne z \mid T=t] &\le \eps/16 \, , \label{ap-prp:low-err-rho}
  \end{align}
  where $(X,Y) \sim \rho$, $T = P(X,Y)$, and $\sigma_t = (\rho \mid T = t)$.
\end{lemma}
\begin{proof}
  Suppose condition~\eqref{ap-prp:low-err-rho} holds. Then
  Lemma~\ref{lem:app-main} implies that there exists a
  rectangle $L$ such that $\sigma_t(L)\geq 9/16$ and $\Pr[g(X,Y) \ne z \mid
  (X,Y) \in L] \le \eps$. The latter condition may be rewritten as
  $\rho(L \setminus (g^{-1}(z)\cup g^{-1}(*))) \le \eps\rho(L)$, i.e., $L$ is
  $\eps$-error $z$-monochromatic for $g$ under $\rho$.

  Suppose~\eqref{ap-prp:low-kl-div} also holds. Then, by the Substate
  Theorem, for every subset $S\ceq\cX\times\cY$, we have
  \[
    \rho(S) \ge \frac{\sigma_t(S)}{2^{2 + 2/\sigma_t(S) + 2d/\sigma_t(S)}} \, ,
  \]
  where $d = \dkl{\sigma_t}{\rho}$. Taking $S$ to be the above rectangle $L$,
  and noting that $\sigma_t(L) \ge 1/2$, we have
  \[
    \rho(L) \ge \frac{1}{2^{7+4d}} > \frac{1}{2^{\cb^{z,\rho}_\eps(g)}} \, .
  \]
  Since $L$ is $\eps$-error $z$-monochromatic, the definition of the
  corruption bound tells us that $\cb^{z,\rho}_\eps(g) \le
  -\log\rho(L)$, which contradicts the above inequality.
\end{proof}
\begin{proof}[Proof of Theorem~\ref{ap-thm:scb-ic}]
  Suppose, to the contrary, that $\IC^\rho_\eps(f) \le
  \scb^{z,\rho}_{\eps',\eps}(f)/400 - 1/50$. Let $P^*$ be a protocol for $f$
  achieving the $\eps$-error information cost under $\rho$. By a standard
  averaging argument, we may fix the public randomness of $P^*$ to obtain a
  private-coin protocol $P$ that computes $f$ with error $2\eps$, and has
  $\icost^\rho(P) \le 2\icost^\rho(P^*)$.  Let $g$ be the function achieving
  the maximum in Eq.~\eqref{ap-def:scb}, the definition of the smooth corruption
  bound, with error parameter $\eps'$ and perturbation parameter $\eps$. Then
  $\scb^{z,\rho}_{\eps',\eps}(f) = \cb^{z,\rho}_{\eps'}(g)$ and $P$ computes
  $g$ with error $3\eps \le 1/500$. Furthermore,
  \[
    \rho(g^{-1}(z)) 
    \ge \rho(f^{-1}(z)) - \Pr_{(X,Y)\sim\rho}[f(X,Y) \ne g(X,Y)]
    \ge 3/20 - \eps > 3/20 - 1/500 \, .
  \]
  By Lemma~\ref{ap-lem:low-distortion}, there exists a transcript $t$ of $P$
  satisfying conditions~\eqref{ap-cond:tran-output}, \eqref{ap-cond:low-divergence},
  and~\eqref{ap-cond:low-error}. The right hand side
  of~\eqref{ap-cond:low-divergence} is at most $100\icost^\rho(P^*) <
  (\scb^{z,\rho}_{\eps',\eps}(f) - 7)/4 = (\cb^{z,\rho}_{\eps'}(g) - 7)/4$ and
  the right hand side of~\eqref{ap-cond:low-error} is at most $24\eps \le
  \eps'/16$.

  Therefore, conditions~\eqref{ap-prp:low-kl-div} and~\eqref{ap-prp:low-err-rho} in
  Lemma~\ref{ap-lem:no-t-with-three} are {\em both} satisfied, while $\out(t) =
  z$ and $\rho$ is rectangular, which contradicts that lemma.
\end{proof}

\section{The Information Complexity of Orthogonality and Gap-Hamming} \label{ap-sec:ghd}

We now tackle Theorems~\ref{ap-thm:ghd-ic} and~\ref{ap-thm:ghd-ort}. Since
these results are closely connected with a few recent works, and are both
conceptually and technically interesting in their own right, we begin by
discussing why they take so much additional work.

For the remainder of this paper, $\mu_n$ will denote the uniform distribution
on $\b^n\times\b^n$. We will almost always drop the subscript $n$ and simply
use $\mu$.

\subsection{The Orthogonality Problem}

The first thing to address is why the information complexity of these problems
is not already lower bounded by an existing general result of Barak et
al.~\cite{BarakBCR10}.

\paragraph{The Barak-Braverman-Chen-Rao Approach.~} The protocol compression
technique given by Barak et al.~for {\em rectangular}
distributions relates information complexity under such distributions to
communication complexity in what seems like a near-optimal way. Why then are
we not happy with their result? To understand this, consider a protocol $P$
for $\ort_{1/4,n}$ with communication cost $c$, error $\eps$ (for some sufficiently
small constant $\eps$) and information cost $d$, under the uniform
distribution $\mu$.  Their compression result would compress $P$ to a
$2\eps$-error $\ort$ protocol $P^*$ with
\[
  \cost(P^*) = O\bigg(\frac{d\log (c/\eps)}{\eps^2}\bigg) \, .
\]
By the distributional complexity lower bound for $\ort_{1/4,n}$~\cite{Sherstov11ghd}, we
have $\cost(P^*) = \Omega(n)$. {\em However, this does not imply $d =
\Omega(n)$} or even $d = \Omega(n/\polylog(n))$! In particular, we may have
the weird situation that $d = O(1)$ and $c = 2^{\Omega(n)}$. Thus, our
lower bound for $\IC(\ort_{b,n})$ is in fact a strong result, far from what follows
from prior work.

\paragraph{A Word About Our Approach.~}
Turning to {\em our} proof for a moment, we now see that we need to lower
bound $\cb^\lambda(\ort_{b,n})$ for a {\em rectangular} $\lambda$. We make the
most natural choice, picking $\lambda = \mu$, the uniform input distribution.
Our proof is then heavily inspired by two recent proofs of an optimal
$\Omega(n)$ lower bound on $\RR(\ghd_n)$, namely those of Chakrabarti and
Regev~\cite{ChakrabartiR11}, and Sherstov~\cite{Sherstov11ghd}. At the heart
of our proof is the following anti-concentration lemma, which says that when
pairs $(x,y)$ are randomly drawn from a large rectangle in $\b^n\times\b^n$,
the inner product $\ip xy$ cannot be too sharply concentrated around zero.

\begin{lemma}[Anti-concentration] \label{ap-lem:anti-conc}
  Let $n$ be sufficiently large, let $b \geq 66$ be a constant, 
  and let $\eps=\tail(2.01b)$. Then there exists $\delta > 0$ such that 
  for all $A,B\ceq\b^n$ with $\min\{|A|,|B|\} \ge 2^{n-\delta n}$, we have
  \begin{equation} \label{ap-eq:main-anti-conc}
    \Pr_{(X,Y) \in_R A\times B}\Big[ \ang{X,Y} \notin [-b\sqrt{n}, b\sqrt{n}] 
    \Big] \ge \eps \, ,
  \end{equation}
  where ``$\in_R$'' denotes ``is chosen uniformly at random from''.
\end{lemma}

The proof of this anti-concentration lemma has several technical steps, and we
give this proof in Section~\ref{ap-sec:anti-conc}. Below, we prove
Theorem~\ref{ap-thm:ghd-ic} using this lemma, and then discuss what is new
about this lemma.

\begin{theorem}[Precise restatement of Theorem~\ref{ap-thm:ghd-ic}]
  \label{ap-thm:ghd-ic-precise}
  Let $b\geq 1/5$ be a constant. Then $\cb^{1,\mu}_{\theta}(\ort_{b,n}) =
  \Omega(n)$, for $\theta=\tail(2.01 \max\{66,b\})$. Hence, we have
  $\IC^\mu_{\theta/400}(\ort_{b,n}) = \Omega(n)$.
\end{theorem}
\begin{proof}
  We first estimate the corruption bound.  Let $\delta$ be the constant whose
  existence is guaranteed by Lemma~\ref{ap-lem:anti-conc}. For $b\ge 66$,
  Eq.~\eqref{ap-eq:main-anti-conc} states precisely that
  $\theta\mbox{-mono}^{1,\mu}(\ort_{b,n}) \le 2^{-\delta n}$. Thus, it follows
  that $\cb^{1,\mu}_{\theta}(\ort_{b,n}) \ge \delta n = \Omega(n)$. For
  $b<66$, we note that
  \[
    \Pr_{(X,Y)\in_R A\times B} \left[\ip XY \notin [-b\sqrt n,b\sqrt n] \right] \ge
    \Pr_{(X,Y)\in_R A\times B} \left[\ip XY \notin [-66\sqrt n,66\sqrt n] \right],
  \]
  for any $A,B\subseteq \b^n$. Therefore, using Lemma~\ref{ap-lem:anti-conc}
  as before, we can conclude that $\cb^{1,\mu}_{\theta}(\ort_{b,n}) =
  \Omega(n)$ for $\theta = \tail(2.01\times 66)$.

  To lower bound the information complexity, we first note that
  \[
    \scb^{1,\mu}_{\theta,\theta/400}(\ort_{b,n}) 
    \ge \scb^{1,\mu}_{\theta,0}(\ort_{b,n}) 
    = \cb^{1,\mu}_{\theta}(\ort_{b,n}) 
    = \Omega(n) \, .
  \]
  Since $b\geq 1/5$, standard estimates of the tail of a binomial distribution
  give us that $\mu(\ort_{b,n}^{-1}(1))> 3/20$ for large enough $n$. Further,
  we have $\theta =\tail(2.01\max \{66,b\}) < 1/4$. Applying
  Theorem~\ref{ap-thm:scb-ic}, we conclude that
  $\IC^\mu_{\theta/400}(\ort_{b,n}) = \Omega(n)$.
\end{proof} 

We now address why the approaches in two recent works do not suffice to prove
Lemma~\ref{ap-lem:anti-conc}.

\paragraph{The Sherstov Approach.~}
At first glance, Lemma~\ref{ap-lem:anti-conc} may appear to be essentially
Sherstov's Theorem 3.3, but it is not! Sherstov's theorem is a special case of
ours that fixes $b = 1/4$, and the smallness of that choice is crucial to
Sherstov's proof. In particular, his proof does not work once $b > 1$. In
order to connect $\ort$ to $\ghd$, however, we need this anti-concentration
with $b$ being a {\em large} constant. Looking ahead a bit, this is because we
need the upper bound in Eq.~\eqref{eq:hb-close-ghd} to be tight enough.

The reason that Sherstov's approach requires $b$ to be small is technical, but
here is a high-level overview.  He relies on an inequality of Talagrand (which
appears as~\cite[Fact 2.2]{Sherstov11ghd}) which states that the projection of
a random vector from $\b^n$ onto a linear subspace $V \ceq \R^n$ is sharply
concentrated around $\sqrt{\dim V}$, which is at most $\sqrt{n}$. Once $b >
1$, this sharp concentration works {\em against} his approach and, in
particular, fails to imply anti-concentration of $\ip XY$ in $[-b\sqrt n,
b\sqrt n]$, which is now too large an interval.

\paragraph{The Chakrabarti-Regev Approach.~}
At second glance, Lemma~\ref{ap-lem:anti-conc} may appear to be a variant of the
``correlation inequality'' (Theorem 3.5 and Corollary 3.8) of Chakrabarti and
Regev. This is true to an extent, but crucially our lemma is not a corollary
of that correlation inequality, which we state below.

\begin{fact}[Equivalent to Corollary~3.8 of~\cite{ChakrabartiR11}]
\label{ap-fact:cr-correl}
  Let $n$ be sufficiently large, and let $b > 0$ and $\eps > 0$ be constants.
  Then there exists $\delta > 0$ such that for all
  $A,B\ceq\b^n$ with $\min\{|A|,|B|\} \ge 2^{n-\delta n}$, we have
  \begin{equation} \label{ap-eq:cr-correl}
    \nu_b(A\times B) \ge (1-\eps) \mu(A\times B) \, ,
  \end{equation}
  where $\nu_b = \frac12(\xi_{-2b/\sqrt{n}} + \xi_{2b/\sqrt{n}})$ and $\xi_p$
  is the distribution of $(x,y) \in \b^n\times\b^n$ where we pick $x\in_R\b^n$
  and choose $y$ by flipping each coordinate of $x$ independently with
  probability $(1-p)/2$.
\end{fact}

The above is also an anti-concentration statement about inner products in a
large rectangle. One might therefore hope to use it to prove
Lemma~\ref{ap-lem:anti-conc} by showing that one kind of anti-concentration
implies the other for ``counting'' reasons. That is, one might hope that every
large {\em set} $S\ceq\b^n\times\b^n$ that satisfies an inequality
like~\eqref{ap-eq:cr-correl} also satisfies one like~\eqref{ap-eq:main-anti-conc}.

But this is not the case.  Consider the set $S = S_0\cup S_{2b}$ where $S_0$
is any subset of $2^{2n-\delta n}$ inputs such that for all $(x,y)\in S_0$ we
have $\ip xy =0$, and $S_{4b}$ is any subset of $(\eps/2)|S_0|$ inputs such
that for all $(x,y)\in S_{4b}$ we have $\ip xy = 4b\sqrt n$. Then, by
construction, we have $\Pr_{(x,y)\in_R S}\big[\ip xy \notin [-b\sqrt n, b\sqrt
n]\big] \le \eps/2 < \eps$, so $S$ does not satisfy an inequality
like~\eqref{ap-eq:main-anti-conc}.  However, for several choices of $\eps$ and
$b$, it does satisfy the analogue of inequality~\eqref{ap-eq:cr-correl}: a short
calculation shows that $\nu_b(S) \ge \frac12\xi_{2b/\sqrt n}(S_{4b}) \ge
\frac12\eps\,e^{5b^2}\mu(S) \ge \mu(S)$. 

Thus, even given Fact~\ref{ap-fact:cr-correl}, we still need to use the
rectangularity of $S$ to prove Lemma~\ref{ap-lem:anti-conc}. It is this need to
use rectangularity carefully that leads to the longish technical proof to
follow, in Section~\ref{ap-sec:anti-conc}.

\subsection{The Gap-Hamming Problem}
\label{subsec:cnt}

We now address the issue of proving a strong lower bound on $\IC^\mu(\ghd)$.
As before, we first note why existing methods do not imply an $\Omega(n)$
lower bound, and then give our approach. We stress that our approach is, at
this point, a {\em program only} and stops short of settling
Conjecture~\ref{conj:ghd-ic}, i.e., proving that $\IC^\mu(\ghd) = \Omega(n)$.

\paragraph{Previous Approaches.~} The orthogonality problem $\ort$ is
intimately related to the Gap-Hamming Distance problem $\ghd$. This was first
noted by Sherstov, who used an ingenious technique to prove that $\RR(\ghd_n)
= \Omega(n)$ based on his lower bound $\RR(\ort_{1/4,n}) = \Omega(n)$. He gave
a reduction from $\ort$ to $\ghd$ wherein a protocol for $\ghd$ was called
{\em twice} to obtain a protocol for $\ort$. But this style of reduction does
not yield a relation between {\em information} complexities, and so the lower
bound on $\IC^\mu(\ort)$ in Theorem~\ref{ap-thm:ghd-ic-precise} does not
translate into a lower bound on $\IC^\mu(\ghd)$.

The Chakrabarti-Regev proof~\cite{ChakrabartiR11} of the same bound
$\RR(\ghd_n) = \Omega(n)$ introduces a technique that they call
corruption-with-jokers which in turn is subsumed by what Jain and
Klauck~\cite{JainK10} have called the ``smooth rectangle bound.'' In fact,
Jain and Klauck define two variants of the smooth rectangle bound: a
linear-programming-based variant that they denote $srec$, and a ``natural''
variant that they denote $\widetilde{srec}$.  It is the former variant that
subsumes the Chakrabarti-Regev technique, whereas our work here corresponds to
the latter variant.

Jain and Klauck do give a pair of translation lemmas, showing that the two
variants are asymptotically equivalent up to some changes in parameters.
Therefore, the Chakrabarti-Regev approach does yield a lower bound on
$\scb^\lambda(\ghd_n)$, but the distribution $\lambda$ that comes out of
applying the appropriate translation lemma is non-rectangular. Therefore, we
cannot apply Theorem~\ref{ap-thm:scb-ic}.

Furthermore, even granting Conjecture~\ref{ap-conj:main} (as claimed by
Kerenidis et al.~\cite{Kerenidis+12}), this line of reasoning will only lower
bound $\IC^\lambda(\ghd)$ for an artificial distribution $\lambda$, and will
not lower bound $\IC^\mu(\ghd)$.

\paragraph{Our Approach.~} 
Our idea is that, for large $b$, the function $\ghd_n$ is at least as ``hard''
as a function that is ``close'' to $\ort_{b,n}$, under a uniform input
distribution. To be precise, we have the following connection between $\ghd$
and $\ort$. Recall that $\mu_n$ is the uniform distribution on
$\b^n\times\b^n$.

\begin{theorem}[Precise restatement of Theorem~\ref{ap-thm:ghd-ort}]
  \label{ap-thm:ghd-ort-precise}
  Let $n$ be sufficiently large, let $b\ge 100$ be a constant, and let
  $\tail(1.99b)\le \theta \le 1/1600$. Let $n'= n + \frac12(1.99b-1)\sqrt n$.  Then, we have
  \[
    \scb^{1,\mu_n}_{400\theta, \theta}(\ghd_n) 
    = \Omega(\cb_{400\theta}^{1,\mu_{n'}}(\ort_{b,n'})) - O(\sqrt n) \, .
  \]
  Combining this with Theorem~\ref{ap-thm:scb-ic}, we then have
  $\IC^{\mu_n}_\theta(\ghd_n)=\Omega(\cb_{400\theta}^{1,\mu_{n'}}(\ort_{b,n'}))
  -O(\sqrt n)$.
\end{theorem}
\begin{proof}[Remark]  
  Suppose we could strengthen Theorem~\ref{ap-thm:ghd-ic-precise} by changing
  the constant $2.01$ in that theorem to $1.98$, i.e., suppose we had
  $\cb^{1,\mu}_{\eps}(\ort_{b,n}) = \Omega(n)$ with $\eps = \tail(1.98b)$.
  Then the present theorem would give us $\IC^\mu_{\eps/400}(\ghd_n) = 
  \Omega(n)$, since $\eps/400 > \tail(1.99b)$ for large enough $b$.
  \renewcommand{\qedsymbol}{}
\end{proof}
\begin{proof}  
  Put $t = n' - n = \frac12(1.99b-1)\sqrt{n}$. Consider the padding $(x,y)
  \in\b^n \longmapsto (x',y') \in\b^{n'}$ defined by $x' = (1,1,\ldots,1,x)$
  and $y' = (-1,-1,\ldots,-1,y)$.  Then we have $\ip {x'}{y'} = \ip xy - t$.
  Since $b \ge 100$, for $b' := 1.99b$, we have
  \begin{equation}
   \label{eq:padding}
     \ip xy \in \big[-\sqrt n, b'\sqrt n\big] \implies
     \ip{x'}{y'} \in \big[-b\sqrt{n'}, b\sqrt{n'}\big] \, .
  \end{equation} 
  Let $h:\b^n\times\b^n\to\b$ be the partial function defined as follows:
  \[
    h(x,y)= \begin{cases}
      \ghd_n(x,y), &\text{if~} \ip xy \le b'\sqrt n \, , \\
      -\ghd_n(x,y), &\text{if~} \ip xy > b'\sqrt n \, . \\
    \end{cases}
  \]
  From~\eqref{eq:padding} and the definition of $\ort$ we can conclude that 
  $\ort_{b,n'}(x',y')\not=1 \implies h(x,y) \notin \{1,*\}$ for all $x,y\in\b^n$.
  Thus, for any rectangle $R\subseteq \b^n\times\b^n$, we have
  \[
    \frac{|(x,y)\in R: h(x,y) \notin \{1,*\}|}{|R|} \ge \frac{|(x',y')\in R': \ort_{b,n'}
    (x',y') \not=1|}{|R'|} \, ,
  \]
  where $R'\subseteq\b^{n'}\times\b^{n'}$ is the rectangle obtained by padding
  each $(x,y)\in R$ as above. Therefore, if $R$ is $\eps$-error
  1-monochromatic for $h$ under $\mu_n$, then $R'$ is $\eps$-error
  1-monochromatic for $\ort_{b,n'}$ under $\mu_{n'}$. Hence,
  $\epsmono^{1,\mu_n}(h) \le 2^{2t}\epsmono^{1,\mu_{n'}} (\ort_{b,n'})$ and
  thus, $\cb^{1,\mu_n}_\eps(h) \ge
  \cb^{1,\mu_{n'}}_\eps\left(\ort_{b,n'}\right)-2t$.

  By standard estimates of the tail of a binomial
  distribution~\cite{Feller-book}, we have
  \begin{equation} \label{eq:hb-close-ghd}
    \Pr_{(X,Y)\sim\mu_n}[h(X,Y) \ne \ghd_n(X,Y)] 
    = \Pr_{(X,Y)\sim\mu_n}[\ip XY > b'\sqrt n] 
    \le \tail(b') = \tail(1.99b) \, .
  \end{equation}
  Therefore, $\scb^{1,\mu_n}_{\eps,\theta}(\ghd_n) \ge \cb^{1,\mu_n}_\eps(h)\ge
  \cb^{1,\mu_{n'}}_{\eps}(\ort_{b,n'}) -2t$ with $\theta \ge
  \tail(1.99b)$. The proof is now completed by applying Theorem~\ref{ap-thm:scb-ic}:
  for the setting $\eps=400\theta$, we have
  $0 \le 384\theta \le \eps < 1/4$ and $\mu_n(\ghd_n^{-1}(1)) \ge 3/20$. 
  Therefore, we can conclude
  \[
    \IC^\mu_\theta(\ghd_n) 
    = \Omega\big(\scb^{1,\mu_n}_{\eps,\theta}(\ghd_n)\big) - O(1)
    = \Omega(\cb^{1,\mu_{n'}}_\eps(\ort_{b,n'})) -O(\sqrt n) \, . \qedhere
  \]
\end{proof}

\section{Proof of the Anti-Concentration Lemma} \label{ap-sec:anti-conc}

Finally, we turn to the most technical part of this work: a proof of our new
anti-concentration lemma, stated as Lemma~\ref{ap-lem:anti-conc} earlier.

\subsection{Preparatory Work and Proof Overview}

Let us begin with some convenient notation. We denote the (density function of
the) standard normal distribution on the real line $\R$ by $\gamma$. We also
denote the standard $n$-dimensional Gaussian distribution by $\gamma^n$. For a
set $A \ceq \R^n$, we denote by $\gamma^n|_{A}$ the distribution $\gamma^n$
conditioned on belonging to $A$. For a distribution $P$ on $\R^n$, we define
its ``distance to Gaussianity'', denoted $\DD_\gamma(P)$ as follows.
\[
  \DD_\gamma(P) = \dd{P}{\gamma^n}
  := \int P(x) \ln\frac{P(x)}{\gamma^n(x)} dx \, .
\]
The latter quantity is the well-known relative entropy for continuous
probability distributions, and is the analogue of $\DD_{\mathrm{KL}}$, which
we have used earlier. Note that the logarithm here is to the base $e$, and not
$2$ as it was earlier.

Let $X,Y$ be possibly correlated random variables, with density functions
$P_X$ and $P_Y$ respectively. Let $P_{X\mid Y=y}$ denote the conditional
probability density function of $X$ given the value $y$ of $Y$. We will
sometimes write $\DD_\gamma(X)$ as shorthand for $\DD_\gamma(P_X)$, and we
will define
\[
  \DD_\gamma(X\mid Y) = \E_y[ \dd{P_{X\mid Y=y}}{\gamma} ] \, .
\]

For a vector $x\in\R^n$ and a linear subspace $V\ceq\R^n$, we denote the
orthogonal projection of $x$ onto $V$ by $\proj_V x$. We denote the Euclidean
norm of $x$ by $\|x\|$.

\paragraph{The Setup.~}
For a contradiction, we begin by assuming the negation of
Lemma~\ref{ap-lem:anti-conc}. That is, we assume that there is a constant
$b\ge 66$ such that for all constants $\delta
> 0$, there exist $A,B\ceq\b^n$ such that
\begin{equation}\label{ap-eq:largerec}
\min\{|A|,|B|\} \ge 2^{n-\delta n} \text{ and }
\end{equation}
\begin{equation} \label{ap-eq:conc}
  \Pr_{(X,Y) \in_R A\times B}\Big[ \ang{X,Y} \notin [-b\sqrt{n}, b\sqrt{n}] 
  \Big] < \eps := \tail(2.01b) \, .
\end{equation}

We treat the sets $A$ and $B$ asymmetrically in the proof. Using the largeness
of $A$, and appealing to a concentration inequality of Talagrand, we identify
a subset $V\ceq A$ consisting of $\Theta(n)$ vectors such that
\begin{enumerate}
  \item[(P1)] the vectors in $V$ are, in some sense, near-orthogonal; and
  \item[(P2)] the quantity $\ip xY$, where $y\in_R B$, is concentrated around zero
  for {\em each} $x\in V$, in the sense of~\eqref{ap-eq:conc}.
\end{enumerate}
This step is a simple generalization of the first part of Sherstov's argument
in his proof that $\RR(\ghd_n) = \Omega(n)$.

As for the set $B$, we consider its {\em Gaussian analogue} $\widetilde B :=
\{\tilde y \in \R^n: \sign(\tilde y)\in B\}$. Consider the random variable
$Q_x = \ip{x}{\widetilde Y}/\sqrt n$, for an arbitrary $x\in V$ and
$\widetilde Y\sim \gamma^n|_{\widetilde{B}}$. On the one hand, we can show
that property~(P2) above implies ``concentration'' for $Q_x$ in some sense.
Combined with property~(P1), we have that projections of the set $\widetilde
B$ along $\Omega(n)$ near-orthogonal directions are all ``concentrated.'' On
the other hand, arguing along the lines of Chakrabarti-Regev, we cannot have
too much concentration along so many near-orthogonal directions, because
$\widetilde B$ is a ``large'' subset of $\R^n$. The incompatibility of these
two behaviors of $Q_x$ gives us our desired contradiction.

It remains to identify a suitable notion of ``concentration'' that lets us
carry out the above program.  The notion we choose is the {\em escape
probability} $p^* = \Pr[|Q_x| > (c+\alpha)b]$, for suitable constants
$c,\alpha > 0$ that we shall determine later.


\subsection{The Actual Proof} \label{ap-subsec:anti-conc-proof}

Let $Y$ denote a uniformly distributed vector in $B$.  Define the set
\begin{equation} \label{ap-eq:c-def}
  C := \{x\in A: \Pr_{Y\in_R B}[\ip xY \notin [-b\sqrt n, b\sqrt n]] 
    < 2\eps\} \, . 
\end{equation}
By Eq.~\eqref{ap-eq:conc} and Markov's inequality, we have $|C| \ge \frac12|A|
\ge 2^{n-\delta n -1}$. We now use some geometry.

\begin{fact}[Generalization of {\cite[Lemma 3.1]{Sherstov11ghd}}] \label{ap-fact1}
  Let $\delta>0$ be a sufficiently small constant and let $n$ be large enough.
  Put $k = \ceil{\sqrt\delta n}$.  Suppose $C \ceq \b^n$ has size $|C| \ge
  2^{n-\delta n -1}$. Then there exist $x_1,\ldots,x_k\in A$ such that
  \begin{equation} \label{ap-eq:near-ort}
    \forall\, j\in \{1,\ldots,k\}, ~\text{we have}~~ 
    \| \proj_{\spn\{x_1,x_2,\ldots,x_{j-1}\}}x_j \| \le 2\delta^{1/4}\sqrt{n}
    \, ,
  \end{equation}
\end{fact}
\begin{proof}
  Having chosen $x_1,\ldots,x_{j-1}$ (where $j \le k$), we apply the appropriate
  variant of Talagrand's concentration
  inequality~\cite[Theorem~7.6.1]{AlonSpencer-book} to obtain that $\|
  \proj_{\spn\{x_1,\ldots,x_{j-1}\}}x_j \|$ is sharply concentrated around
  $\sqrt{\dim\spn\{x_1,\ldots,x_{j-1}\}} \le \sqrt k$. In particular, there is
  an absolute constant $c$ such that
  \[
    \Pr_{x_j \in_R \b^n}\left[ 
      \| \proj_{\spn\{x_1,\ldots,x_{j-1}\}}x_j \| > 2\delta^{1/4}\sqrt{n}
    \right] \le 2^{-c\sqrt\delta n} \, .
  \]
  On the other hand $\Pr_{x \in_R \b^n}[x\in A] \ge 2^{-\delta n -1}$, which is
  larger than the above estimate if $\delta$ is sufficiently small. Therefore,
  we can pick a suitable $x_j$ to continue.
\end{proof}

From now on, fix the ``near-orthogonal'' set of vectors $x_1,\ldots,x_k$, with
$k = \ceil{\sqrt\delta n}$, given by Fact~\ref{ap-fact1}.  Recall that
$\widetilde B := \{\tilde y \in \R^n: \sign(\tilde y)\in B\}$. We define a
random variable $\widetilde Y$ correlated with $Y$ as follows. Let
$(Y_1,\ldots,Y_n)$ be the coordinates of $Y$; then define $\widetilde Y_j =
Y_j|W_j|$, where $W_j \sim \gamma$ and put $\widetilde Y = (\widetilde Y_1,
\ldots, \widetilde Y_n)$. Notice that the resulting distribution of
$\widetilde Y$ is exactly $\gamma^n|_{\widetilde B}$.  We now define the
random variable $Q_j$ and its escape probability $p_j^*$ as follows.
\[
  Q_j := \frac{\ip{x_j}{\widetilde Y}}{\sqrt n} \, ; \quad
  p_j^* := \Pr\big[ |Q_j| > (c+\alpha)b \big] \, .
\]
We shall eventually fix a particular index $j$ and choose suitable constants $c$ and
$\alpha$ above. As mentioned in the overview, the proof will hinge on a
careful analysis of this escape probability.

\subsubsection*{Lower Bounding the Escape Probability}

We begin the study by showing that there exists an index $j \in
\{1,\ldots,k\}$ such that $Q_j$ behaves quite similarly to a mixture of
shifted standard normal variables (i.e., variances close to $1$, but arbitrary
means). This will in turn yield a lower bound on the corresponding $p_j^*$.

Let $\tilde x_1, \ldots, \tilde x_k$ be the (truly) orthogonal vectors
obtained from $x_1,\ldots,x_k$ by the Gram-Schmidt process, i.e., $\tilde x_i
:= x_i - \proj_{\spn\{x_1,\ldots,x_{i-1}\}} x_i$.  For $i \in \{1,\ldots,k\}$,
put $x_i^* = \tilde x_i/\nm{\tilde x_i}$, and let $x_{k+1}^*, \ldots, x_n^*$ be a
completion of these vectors to an {\em orthonormal} basis of $\R^n$.
Expressing $\widetilde Y$ in this basis, and noting that
$\ip{x_j}{x_i^*} = 0$ for all $i>j$ in step~\eqref{ap-eq:higher-zerop} below, we
derive
\begin{align}
  Q_j
  &= \frac1{\sqrt{n}} \lip {x_j} {\sum_{i=1}^n\ip {\widetilde Y}{x_i^*} x_i^*} \notag \\
  &= \sum_{i=1}^j\frac{\ip {x_j} {x_i^*}}{\sqrt n} \ip {\widetilde Y}{x_i^*} \label{ap-eq:higher-zerop} \\
  &= \frac{\ip {x_j}{x_j^*}}{\sqrt n} \ip{\widetilde Y}{x_j^*} +
    \sum_{i=1}^{j-1}\frac{\ip {x_j}{x_i^*}}{\sqrt n}
      \ip {\widetilde Y}{x_i^*} \notag \\
  &= r_j Z_j + S_j \, ,\label{ap-eq:ip-rzs}
\end{align}
where we define
\[
  r_j := \frac{\ip{x_j}{x_j^*}}{\sqrt n} \, , \quad
  Z_j := \ip{\widetilde Y}{x_j^*} \, , \quad
  S_j := 
    \sum_{i=1}^{j-1}\frac{\ip {x_j}{x_i^*}}{\sqrt n}
      \ip {\widetilde Y}{x_i^*} \, .
\]
The Pythagorean theorem says that
  $\ip{x_j}{x_j^*}^2 = \|x_j\|^2 - \|
  \proj_{\spn\{x_1,x_2,\ldots,x_{j-1}\}}x_j \|^2$.
Recalling that $\|x_j\| = \sqrt{n}$ and using~\eqref{ap-eq:near-ort}, we conclude
that
\begin{equation} \label{ap-eq:rj-bounds}
  \forall\, j\in \{1,\ldots,k\}, ~\text{we have}~~ 
  1 - 4\sqrt\delta \le r_j \le 1 \, .
\end{equation}

\begin{lemma} \label{ap-lem:gauss-like-coord}
  There exists $j\in\{1,\ldots,k\}$ such that $\DD_\gamma(Z_j \mid S_j) \le
  \sqrt\delta$.
\end{lemma}
\begin{proof}
  Since $|B|\ge 2^{n-\delta n}$, we have $\gamma^n(\widetilde{B}) \ge
  2^{-\delta n}$. By definition, we have
  $\DD_{\gamma}(\gamma^{n}|_{\tilde{B}}) = -\ln\gamma^{n}(\tilde{B}) \le
  (\ln2) \delta n \le \delta n$. On the other hand, by the chain rule for
  relative entropy, we have 
  \[
    \DD_{\gamma}(\gamma^{n}|_{\tilde{B}}) 
    = \DD_\gamma(\widetilde Y)
    = \DD_\gamma\left(\ip{\widetilde Y}{x_1^*},\ldots,\ip{\widetilde Y}{x_n^*}\right)
    = \sum_{j=1}^n \DD_\gamma\left(\ip{\widetilde Y}{x_j^*} \mid
        \ip{\widetilde Y}{x_1^*},\ldots,\ip{\widetilde Y}{x_{j-1}^*}\right) \, .
  \]
  Recalling that $k = \ceil{\sqrt{\delta}n}$, we deduce that there exists an
  index $j\in\{1,\ldots,k\}$ such that 
  \begin{equation} \label{ap-eq:indj-gamdiv}
     \DD_\gamma\left(\ip{\widetilde Y}{x_j^*} \mid
      \ip{\widetilde Y}{x_1^*},\ldots,\ip{\widetilde Y}{x_{j-1}^*}\right) 
    \le \sqrt\delta \, .
  \end{equation}
  Since $S_j$ is a function of $\ip{\widetilde Y}{x_1^*},\ldots,\ip{\widetilde
  Y}{x_{j-1}^*}$, we conclude that $\DD_\gamma(Z_j \mid S_j) \le \sqrt\delta$.
\end{proof}

For the rest of our proof, we fix an index $j$ as guaranteed by
Lemma~\ref{ap-lem:gauss-like-coord}. We put $r = r_j, Z = Z_j, S = S_j$, $Q =
Q_j$, and $p^* = p_j^*$. Now define the set
\[
  \cS = \{s\in\R:\, \DD_\gamma(Z \mid S=s) \le \delta^{1/4}\} \, ,
\]
so that $\Pr[S\notin\cS] \le \delta^{1/4}$ by Markov's inequality. Clearly,
either $\Pr[S \ge 0 \mid S\in\cS] \ge \frac12$ or $\Pr[S \le 0 \mid S\in\cS]
\ge \frac12$. In what follows, we shall assume that the former condition
holds; it will soon be clear that this does not lose generality. Under this
assumption we have
\begin{equation} \label{ap-z-conditioned}
  \DD_\gamma(Z \mid S \ge 0 \wedge S\in\cS) \le 2\delta^{1/4} \, .
\end{equation}
Therefore, by Pinsker's inequality~\cite{CoverThomas-book}, the statistical
distance between the distribution $\gamma$ and the distribution of $(Z \mid S
\ge 0 \wedge S\in\cS)$ is at most $\sqrt{2(2\delta^{1/4})} = 2\delta^{1/8}$.
Using this fact below, we get
\begin{align}
  p^* 
  &\ge \Pr[Q > (c+\alpha)b] \notag \\
  &= \Pr[rZ+S \ge (c+\alpha)b \mid S\ge 0 \wedge S\in\cS]\cdot
    \Pr[S\ge 0 \mid S\in\cS]\cdot \Pr[S\in\cS] \notag \\
  &\ge \textstyle \frac12(1-\delta^{1/4}) \Pr[rZ+S \ge (c+\alpha)b 
     \mid S\ge 0 \wedge S\in\cS] \notag \\
  &\ge \textstyle \frac12(1-\delta^{1/4}) \Pr[Z \ge (c+\alpha)b/r 
     \mid S\ge 0 \wedge S\in\cS] \notag \\
  &\ge \textstyle \frac12(1-\delta^{1/4}) 
    \left(\tail((c+\alpha)b/r) - 2\delta^{1/8}\right) \notag \\
  &\ge \frac{1-\delta^{1/4}}{2} \left(\tail\left(
    \frac{(c+\alpha)b}{1-4\sqrt\delta}\right) - 2\delta^{1/8}\right) \, ,
    \label{ap-eq:escape-lb}
\end{align}
where the final step uses the lower bound on $r$ given by~\eqref{ap-eq:rj-bounds}.

\subsubsection*{Upper Bounding the Escape Probability}

Recall that we had fixed a specific index $j$ after the proof of
Lemma~\ref{ap-lem:gauss-like-coord}, and that $Q = Q_j = \ip{x_j}{\widetilde
Y}/\sqrt{n}$.  We shall now explore the relation between $\ip{x_j}{Y}$ and
$\ip{x_j}{\widetilde Y}$ to upper bound the escape probability.  At this point
it would help to review the discussion of the relation between $Y$ and
$\widetilde Y$ at the beginning of Section~\ref{ap-subsec:anti-conc-proof}.

For simplicity, we put $x := x_j$ and assume, w.l.o.g., that
$x=(1,1,\ldots,1)$ so that $\ip{x}{y}=\sum_{i=1}^{n}y_{i}$. This is legitimate
because, if $x_{i}= -1$, we can flip $x_{i}$ to 1 and $y_{i}$ to $-y_{i}$
without changing $\ip{x}{y}$.

Recall that each coordinate $\widetilde{Y}_{i}$ of $\widetilde{Y}$ has the
same distribution as $Y_i|W_i|$, where the variables $\{W_i\}$ are independent
and each $W_i \sim \gamma$. Define $T := \sum_{i=1}^n Y_i/\sqrt{n}$; note that
$T$ is a {\em discrete} random variable.  After some reordering of
coordinates, we can rewrite 
\[
  \sqrt{n}\,Q = \ip{x}{\widetilde{Y}}
  = \left( |W_1| + |W_2| + \cdots + \big|W_{\frac{n}{2}+\frac{T\sqrt{n}}{2}}\big| \right)
    - \left( \big|W_{\frac{n}{2}+\frac{T\sqrt{n}}{2}+1}\big| + \cdots + |W_n| \right) \, .
\]
Each $|W_i|$ has a so-called {\em half normal} distribution. This is a
well-studied distribution: in particular, for each $i$, we know that
\[
  \E\big[\, |W_i| \,\big] = \sqrt{\frac2\pi} \, , \quad 
  \Var\big[\, |W_i| \,\big] = 1 -
  \frac2\pi \, .
\]
Thus, for each value $t$ in the range of $T$, we have $\E[\sqrt{n}\,Q \mid T =
t] = t\sqrt{2n/\pi}$ by linearity of expectation, and $\Var[\sqrt{n}\,Q \mid T =
t] = (1-2/\pi)n$ by the independence of the variables $\{W_i\}$. The
half-normal distribution is well-behaved enough for us to apply Lindeberg's
version of the central limit theorem~\cite{Feller-book}: doing so tells us that
as $n$ grows, the distribution of
\[
  \frac{\sqrt{n}\,Q^{(t)} -
    \E[\sqrt{n}\,Q^{(t)}]}{\sqrt{\Var[\sqrt{n}\,Q^{(t)}]}}
  = \frac{Q^{(t)} - t\sqrt{2/\pi}}{\sqrt{1-2/\pi}}
\]
converges to $\gamma$, where $Q^{(t)} = (Q \mid T=t)$. In other words, the
distribution of $Q^{(t)}$ converges to the (shifted and scaled) normal
distribution $\mathcal{N}(t\sqrt{2/\pi}, 1-2/\pi)$. Therefore, the
distribution of $Q$ converges to a mixture of such distributions.
Fix the constants
\[
  c := \sqrt{2/\pi} \, ; \quad
  \sigma := \sqrt{1-2/\pi} \, .
\]
Then the distribution of $Q$ converges to that of $V + cT$,
where $V \sim \mathcal{N}(0,\sigma^2)$ is independent of $T$.
Using the convergence, we can easily prove the following claim.

\begin{claim} \label{ap-lem:exp-coshbxn}
  For sufficiently large $n$, we have
  $p^* = \Pr\big[ |Q| > (c+\alpha)b \big]
  \le 2\Pr\big[ |V + cT| > (c+\alpha)b \big]$.
  \qed
\end{claim}


Recalling that $x\in C$, and using~\eqref{ap-eq:c-def}, we have $\Pr\big[|T| >
b\big] \le 2\eps$. This lets us upper bound $p^*$ as follows.
\begin{align}
  \frac{p^*}{2}
  &\le \Pr\Big[ |V + cT| > (c+\alpha)b ~\Big|~ |T| \le b \Big] + \Pr\big[ |T| > b \big] \notag \\
  &\le \Pr\big[ |V| > \alpha b \big] + \Pr\big[ |T| > b \big] \notag \\
  &\le 2 \Pr[ V/\sigma > (\alpha/\sigma)b ] + 2\eps \notag \\
  &= 2\tail((\alpha/\sigma)b) + 2\tail(2.01b) \, , \label{ap-eq:escape-ub}
\end{align}
where in the last step we use the definition of $\eps$ as given in~\eqref{ap-eq:conc}.

\subsubsection*{Completing the Proof}

To complete the proof of the anti-concentration lemma, we combine the lower
bound~\eqref{ap-eq:escape-lb} with the upper bound~\eqref{ap-eq:escape-ub} to obtain
\[
  \frac{1-\delta^{1/4}}{2} \left(\tail\left(
    \frac{(c+\alpha)b}{1-4\sqrt\delta}\right) - 2\delta^{1/8}\right) 
  \le 4\tail\left(\frac{\alpha b}{\sigma}\right) + 4\tail(2.01b) \, .
\]
Recall that we had started by assuming the negation of
Lemma~\ref{ap-lem:anti-conc}, in Eqs.~\eqref{ap-eq:largerec}
and~\eqref{ap-eq:conc}. Thus, the above inequality is supposed to hold for
some constant $b\ge 66$ and all constants $\delta > 0$. However, if set
$\alpha = 2.01\sigma$, we can get a contradiction: as $\delta\to 0$, the
left-hand side approaches $\frac12 \tail((c+2.01\sigma)b)$, whereas the
right-hand side is $8\tail(2.01b)$.  Plugging in the values of $c$ and
$\sigma$, we note that $c+2.01\sigma < 2.01$. Therefore, if we choose
$\delta$ small enough, we have a contradiction.

\section{Acknowledgment}

We are grateful to Ryan O'Donnell for an important technical discussion and some clarifications.

\eat{
Let $g$ be the function defined as follows:
\begin{align*}
g(x,y) 
&=f(x,y) &&\text{ if } \ip{x}{y} \le b\sqrt{n}\\
&=1-f(x,y) && \text{ if } \ip{x}{y} > b\sqrt{n} \, ,
\end{align*}
for some constant $b>1$ that we will fix later.

Assume that $\text{Corr}^{1,\mu}_{\eps_{1}}(g) > 2^{-\delta_{1} n}$ for any constant $\delta_{1} >0$.
Using the definition of corruption bound, we can find two large sets $A,B \subseteq \b^n$
such that $\mu(A) \ge 2^{-\delta_{1} n}$, $\mu(B) \ge 2^{-\delta_{1} n}$ and
\begin{equation}\label{ap-correqu}
   \Pr_{x\in A, y\in B}[ g(x,y) = 0]=\Pr_{x\in A, y\in B} \left[ \ip{x}{y}
\notin [0,b\sqrt{n}] \right] < \eps_{1}.
\end{equation}

}

\bibliographystyle{alpha}
\bibliography{../../super}

\newcommand{\etalchar}[1]{$^{#1}$}
\begin{thebibliography}{CSWY01}

\bibitem[Abl96]{Ablayev96}
Farid Ablayev.
\newblock Lower bounds for one-way probabilistic communication complexity and
  their application to space complexity.
\newblock {\em Theoretical Computer Science}, 175(2):139--159, 1996.

\bibitem[AS00]{AlonSpencer-book}
Noga Alon and Joel~H. Spencer.
\newblock {\em The Probabilistic Method}.
\newblock Wiley-Interscience, New York, NY, 2000.

\bibitem[BBCR10]{BarakBCR10}
Boaz Barak, Mark Braverman, Xi~Chen, and Anup Rao.
\newblock How to compress interactive communication.
\newblock In {\em Proc. 41st Annual ACM Symposium on the Theory of Computing},
  pages 67--76, 2010.

\bibitem[BJKS04]{BarYossefJKS04}
Ziv {Bar-Yossef}, T.~S. Jayram, Ravi Kumar, and D.~Sivakumar.
\newblock An information statistics approach to data stream and communication
  complexity.
\newblock {\em J. Comput. Syst. Sci.}, 68(4):702--732, 2004.

\bibitem[BPSW06]{BeamePSW06}
Paul Beame, Toniann Pitassi, Nathan Segerlind, and Avi Wigderson.
\newblock A strong direct product theorem for corruption and the multiparty
  communication complexity of disjointness.
\newblock {\em Comput. Complexity}, 15(4):391--432, 2006.

\bibitem[BW11]{BravermanW11}
Mark Braverman and Omri Weinstein.
\newblock A discrepancy lower bound for information complexity.
\newblock Technical Report TR11-164, ECCC, 2011.

\bibitem[CR11]{ChakrabartiR11}
Amit Chakrabarti and Oded Regev.
\newblock An optimal lower bound on the communication complexity of {\sc
  gap-hamming-distance}.
\newblock In {\em Proc. 43rd Annual ACM Symposium on the Theory of Computing},
  pages 51--60, 2011.

\bibitem[CSWY01]{ChakrabartiSWY01}
Amit Chakrabarti, Yaoyun Shi, Anthony Wirth, and Andrew~C. Yao.
\newblock Informational complexity and the direct sum problem for simultaneous
  message complexity.
\newblock In {\em Proc. 42nd Annual IEEE Symposium on Foundations of Computer
  Science}, pages 270--278, 2001.

\bibitem[CT06]{CoverThomas-book}
Thomas~M. Cover and Joy~A. Thomas.
\newblock {\em Elements of Information Theory}.
\newblock Wiley-Interscience [John Wiley \& Sons], Hoboken, NJ, second edition,
  2006.

\bibitem[Fel68]{Feller-book}
William Feller.
\newblock {\em An Introduction to Probability Theory and its Applications}.
\newblock John Wiley, New York, NY, 1968.

\bibitem[IW03]{IndykW03}
Piotr Indyk and David~P. Woodruff.
\newblock Tight lower bounds for the distinct elements problem.
\newblock In {\em Proc. 45th Annual IEEE Symposium on Foundations of Computer
  Science}, pages 283--289, 2003.

\bibitem[JK10]{JainK10}
Rahul Jain and Hartmut Klauck.
\newblock The partition bound for classical communication complexity and query
  complexity.
\newblock In {\em Proc. 25th Annual IEEE Conference on Computational
  Complexity}, pages 247--258, 2010.

\bibitem[JKS03]{JayramKS03}
T.~S. Jayram, Ravi Kumar, and D.~Sivakumar.
\newblock Two applications of information complexity.
\newblock In {\em Proc. 35th Annual ACM Symposium on the Theory of Computing},
  pages 673--682, 2003.

\bibitem[JPY12]{JainPY12}
Rahul Jain, Attila Pereszl\'{e}nyi, and Penghui Yao.
\newblock A direct product theorem for bounded-round public-coin randomized
  communication complexity.
\newblock CoRR abs/1201.1666, 2012.

\bibitem[JRS03]{JainRS03icalp}
Rahul Jain, Jaikumar Radhakrishnan, and Pranab Sen.
\newblock A direct sum theorem in communication complexity via message
  compression.
\newblock In {\em Proc. 30th International Colloquium on Automata, Languages
  and Programming}, pages 300--315, 2003.

\bibitem[JRS09]{JainRS09}
Rahul Jain, Jaikumar Radhakrishnan, and Pranab Sen.
\newblock A property of quantum relative entropy with an application to privacy
  in quantum communication.
\newblock {\em J. ACM}, 56(6), 2009.

\bibitem[Kla03]{Klauck03}
Hartmut Klauck.
\newblock Rectangle size bounds and threshold covers in communication
  complexity.
\newblock In {\em Proc. 18th Annual IEEE Conference on Computational
  Complexity}, pages 118--134, 2003.

\bibitem[KLL{\etalchar{+}}12]{Kerenidis+12}
Iordanis Kerenidis, Sophie Laplante, Virginie Lerays, J\'{e}r\'{e}mie Roland,
  and David Xiao.
\newblock Lower bounds on information complexity via zero-communication
  protocols and applications.
\newblock Technical Report TR12-038, ECCC, 2012.

\bibitem[KN97]{KushilevitzNisan-book}
Eyal Kushilevitz and Noam Nisan.
\newblock {\em Communication Complexity}.
\newblock Cambridge University Press, Cambridge, 1997.

\bibitem[LS09]{LinialS09}
Nati Linial and Adi Shraibman.
\newblock Lower bounds in communication complexity based on factorization
  norms.
\newblock {\em Rand. Struct. Alg.}, 34(3):368--394, 2009.
\newblock Preliminary version in \em Proc. 39th Annual ACM Symposium on the
  Theory of Computing\em\/, pages 699--708, 2007.

\bibitem[LSv08]{LeeSS08}
Troy Lee, Adi Shraibman, and Robert \v{S}palek.
\newblock A direct product theorem for discrepancy.
\newblock In {\em Proc. 23rd Annual IEEE Conference on Computational
  Complexity}, pages 71--80, 2008.

\bibitem[Sha03]{Shaltiel03}
Ronen Shaltiel.
\newblock Towards proving strong direct product theorems.
\newblock {\em Comput. Complexity}, 12(1--2):1--22, 2003.

\bibitem[She08]{Sherstov08}
Alexander~A. Sherstov.
\newblock The pattern matrix method for lower bounds on quantum communication.
\newblock In {\em Proc. 40th Annual ACM Symposium on the Theory of Computing},
  pages 85--94, 2008.

\bibitem[She11a]{Sherstov11ghd}
Alexander~A. Sherstov.
\newblock The communication complexity of gap hamming distance.
\newblock Technical Report TR11-063, ECCC, 2011.

\bibitem[She11b]{Sherstov11dirprod}
Alexander~A. Sherstov.
\newblock Strong direct product theorems for quantum communication and query
  complexity.
\newblock In {\em Proc. 43rd Annual ACM Symposium on the Theory of Computing},
  pages 41--50, 2011.

\bibitem[SS02]{SaksS02}
Michael Saks and Xiaodong Sun.
\newblock Space lower bounds for distance approximation in the data stream
  model.
\newblock In {\em Proc. 34th Annual ACM Symposium on the Theory of Computing},
  pages 360--369, 2002.

\bibitem[Vid11]{Vidick11}
Thomas Vidick.
\newblock A concentration inequality for the overlap of a vector on a large
  set, with application to the communication complexity of the
  gap-hamming-distance problem.
\newblock Technical Report TR11-051, ECCC, 2011.

\bibitem[WZ11]{WoodruffZ11}
David~P. Woodruff and Qin Zhang.
\newblock Tight bounds for distributed functional monitoring.
\newblock Technical Report, available at \url{http://arxiv.org/abs/1112.5153},
  2011.

\end{thebibliography}

\end{document}